\documentclass[letterpaper, 10 pt, conference]{ieeeconf}

\IEEEoverridecommandlockouts
\let\labelindent\relax
\usepackage{enumitem}
\usepackage{amsmath,amssymb,amsthm}
\usepackage{bm}
\usepackage{mathtools}
\usepackage{graphicx}
\usepackage{cite}
\usepackage{xcolor}
\usepackage{algorithm}
\usepackage{algpseudocode}

\usepackage{fancyhdr}

\usepackage{epstopdf} 
\usepackage{style}
\usepackage{dsfont}
\usepackage{hyperref}
\hypersetup{
	colorlinks,
	linkcolor={blue!70!black},
	citecolor={blue!70!black},
	urlcolor={blue!70!black}
}
\allowdisplaybreaks
\usepackage{tikz}
\usetikzlibrary{calc,patterns,decorations.pathmorphing,decorations.markings,decorations.pathreplacing,shapes.geometric,arrows.meta,positioning,fit,backgrounds}

\newtheorem{remark}{\indent Remark}

\title{Nonlinear Moving-Horizon Estimation\\ Using State- and Control-Dependent Models}

\author{Mohammadreza Kamaldar$^{1}$
	\thanks{$^{1}$M. Kamaldar is with the Department of Mechanical, Aerospace \& Biomedical Engineering, University of South Alabama, Mobile, AL 36688, USA.
		{\tt\small mkamaldar@southalabama.edu}}%
}
\begin{document}
	\maketitle
	\thispagestyle{empty}
	
	\begin{abstract}
		This paper presents a state- and control-dependent moving-horizon estimation (SCD-MHE) algorithm for nonlinear discrete-time systems. The nonlinear dynamics and measurement maps are written exactly in pseudo-linear form using state- and control-dependent coefficient (SCDC) matrices. At each time step, the moving-horizon estimation problem is solved by a sequence of sparse quadratic programs, where the SCDC matrices are refrozen along the trajectory computed by the preceding iteration; the estimator requires no Jacobians and retains the nonlinear model exactly. We show that each quadratic program has a unique solution, characterize the fixed points of the iteration, establish geometric convergence under a contraction condition, and prove that the estimation error of the computed estimate is uniformly bounded under uniform observability, bounded disturbances, a bounded arrival-cost error, and iterates confined to a compact set, for all finite iteration counts. In a quadrotor benchmark with a saturating rangefinder, the altitude RMSE of SCD-MHE is 18 times smaller than that of a nonlinear moving-horizon estimator that solves the nonlinear program and 58 times smaller than that of the extended and unscented Kalman filters, and its per-step computation time is 34 times smaller than that of the nonlinear estimator.
	\end{abstract}
	
	
	\section{Introduction}
	
	Feedback control requires an estimate of the state \cite{Simon2006}. For linear systems driven by Gaussian noise, the Kalman filter is the optimal Bayesian estimator, and the extended Kalman filter (EKF) extends it to nonlinear systems by local linearization \cite{Jazwinski1970}. The Jacobians discard the higher-order terms of the model, and the filter can diverge \cite{reif1999extended}. The unscented Kalman filter (UKF) replaces the Jacobians with sigma points that are propagated through the nonlinear maps \cite{Wan2000}.
	
	Moving-horizon estimation (MHE) instead solves a constrained optimization problem over a sliding window of recent measurements \cite{Rao2003, Alessandri2008}. The problem accommodates constraints, and an arrival cost summarizes the data that have left the window; the stability analysis of the estimator rests on this arrival cost \cite{Rao2003, rawlings2017model}. For a nonlinear model, the MHE problem is a nonconvex nonlinear program that is solved at each sampling instant. Real-time iteration schemes perform one Newton-type step per sample \cite{diehl2002real} and require the Jacobians of the model.
	
	This paper develops a state- and control-dependent moving-horizon estimator (SCD-MHE). The nonlinear dynamics and measurement maps are written exactly in pseudo-linear form using state- and control-dependent coefficient (SCDC) matrices. This parameterization originates in state-dependent Riccati equation (SDRE) control \cite{cimen2008sdre}. Here it converts the nonlinear MHE problem into a sequence of sparse quadratic programs: at each time step, the SCDC matrices are evaluated along the trajectory computed by the preceding iteration, the quadratic program is solved, and the trajectory is refined, as in the iterated Kalman filter, which is a Gauss--Newton iteration \cite{BellCathey1993}. No Jacobians are formed, and no term of the model is truncated.
	
	The contributions are threefold. First, we formulate SCD-MHE with warm starting, a displacement-based stopping rule, and a Riccati-type arrival cost update, and we show that each quadratic program has a unique solution. Second, we characterize the fixed points of the iteration, prove geometric convergence under a contraction condition, and prove that the estimation error of the computed estimate is uniformly bounded under uniform observability, bounded disturbances, a bounded arrival-cost error, and iterates confined to a compact set, for all finite iteration counts. Third, on a quadrotor benchmark with a saturating rangefinder, SCD-MHE attains a lower RMSE than the EKF, the UKF, and a nonlinear MHE that solves \eqref{eq:MHSE_cost}, at a per-step computation time 34 times smaller than that of the nonlinear MHE.
	
	\textbf{Notation.} Let $\BBN \triangleq \{0,1,2,\ldots\}$. The $n\times n$ identity matrix is $I_n$, and the zero vector in $\BBR^m$ is $0_m$. For symmetric $S, S_1, S_2 \in \BBR^{n \times n}$, $S \succ 0$ ($S \succeq 0$) indicates that $S$ is positive definite (positive semidefinite), and $S_1 \succ S_2$ ($S_1 \succeq S_2$) indicates that $S_1 - S_2 \succ 0$ ($S_1 - S_2 \succeq 0$). For $x\in\BBR^n$, $\|x\| \triangleq \sqrt{x^\rmT x}$, and, for $W \succ 0$, $\|x\|_W \triangleq \sqrt{x^\rmT W x}$. For $M \in \BBR^{m \times n}$, $\|M\|$ is the induced $2$-norm. For all $x \in \BBR^d$ and all $r \in (0, \infty)$, $\bar{\mathcal{N}}_r(x) \triangleq \{y \in \BBR^d \colon \|y - x\| \le r\}$.
	
	\section{Problem Statement}

	Consider the nonlinear discrete-time system
	\begin{align}
		x_{k+1} & = f(x_k,u_k,k) + w_k,\label{eq_sys1}\\
		y_k & = h(x_k,k) + v_k,\label{eq_meas1}
	\end{align}
	where, for all $k\in\BBN$, $x_k \in \BBR^n$ is the state, $u_k \in \mathcal{U}$ is the known control input, $\mathcal{U} \subset \BBR^m$ is the set of admissible inputs, $y_k \in \BBR^p$ is the measured output, $w_k\in \BBR^n$ is the unknown process disturbance, and $v_k \in\BBR^p$ is the unknown measurement noise. Furthermore, $f\colon \BBR^n \times \BBR^m \times \BBN \to \BBR^n$ and $h\colon \BBR^n \times \BBN\to\BBR^p$ are the state transition and measurement maps. For all $k \in \BBN$, let $Q_k\in\BBR^{n\times n}$ and $R_k\in\BBR^{p\times p}$ be positive definite weighting matrices; in the stochastic interpretation of the cost below, they are the covariances of $w_k$ and $v_k$ \cite{rawlings2017model}, whereas Section~\ref{sec:analysis} treats $w_k$ and $v_k$ as bounded deterministic sequences.
	
	Let $\ell\ge 2$ denote the horizon length. For all $k\ge \ell$, the estimator processes the measurements $y_{k+1-\ell}, \dots, y_k$ and the inputs $u_{k+1-\ell}, \dots, u_{k-1}$, and it estimates the states $x_{k+1-\ell},\ldots,x_k$. Within the window at time step $k$, offset $j \in \{1-\ell, \dots, 0\}$ corresponds to time step $k+j$. Let $n_z \triangleq n\ell + n(\ell-1) + p\ell$, and define the decision vector  $\zeta \in \BBR^{n_z}$ by
	\begin{equation}
		\begin{split}
			\zeta \triangleq \big[ &\chi_{1-\ell}^\rmT \ \cdots \ \chi_{0}^\rmT \ \omega_{1-\ell}^\rmT \ \cdots \ \omega_{-1}^\rmT \\
			&\nu_{1-\ell}^\rmT \ \cdots \ \nu_{0}^\rmT \big]^\rmT, \label{eq:zeta}
		\end{split}
	\end{equation}
	where, for all $j \in \{1-\ell, \dots, 0\}$, $\chi_j \in \BBR^n$ and $\nu_j \in \BBR^p$ are the state and measurement noise variables, and, for all $j \in \{1-\ell, \dots, -1\}$, $\omega_j \in \BBR^n$ is the process disturbance variable. Let $\chi \triangleq [\chi_{1-\ell}^\rmT \ \cdots \ \chi_{0}^\rmT]^\rmT \in \BBR^{n\ell}$ denote the state block of $\zeta$.
	
	For all $k \ge \ell$, let the moving-horizon cost $J_k\colon \BBR^{n_z}\to[0,\infty)$ be defined by
	\begin{align}
		J_k (\zeta) & \triangleq  \|\chi_{1-\ell} - \bar{x}_{k+1-\ell}\|^2_{P_{k+1-\ell}^{-1}} \notag \\
		& \quad + \sum_{j=1-\ell}^{-1} \|\omega_{j}\|^2_{Q_{k+j}^{-1}} + \sum_{j=1-\ell}^{0} \|\nu_{j}\|^2_{R_{k+j}^{-1}},
		\label{eq_JK}
	\end{align}
	where $\bar{x}_{k+1-\ell}\in\BBR^n$ is the arrival cost vector, which summarizes all data processed before the current window, and $P_{k+1-\ell}\in\BBR^{n\times n}$ is the positive definite arrival cost covariance. For all $k \ge \ell$, let $\mathcal{Z}_k \subset \BBR^{n_z}$ denote the set of all $\zeta \in \BBR^{n_z}$ that satisfy
	\begin{align}
		\chi_{j+1} &= f(\chi_j,u_{k+j},k+j) + \omega_j, &&  j \in \{1-\ell, \dots, -1\}, \label{eq:MHSE_dynamics}\\
		y_{k+j} &= h(\chi_j,k+j) + \nu_j, && j \in \{1-\ell, \dots, 0\}. \label{eq:MHSE_measurement}
	\end{align}
	For all $k \ge \ell$, the moving-horizon estimation problem is
	\begin{equation}
		\underset{\zeta \in \mathcal{Z}_k}{\text{minimize}} \ J_k(\zeta). \label{eq:MHSE_cost}
	\end{equation}
	Since \eqref{eq:MHSE_dynamics} and \eqref{eq:MHSE_measurement} are nonlinear in $\chi_j$, \eqref{eq:MHSE_cost} is in general a nonconvex nonlinear program.
	
	
	\section{State- and Control-Dependent Moving-Horizon Estimation}
	\label{sec:scdmhe}
	We avoid the nonconvex program \eqref{eq:MHSE_cost} by writing the nonlinear constraints in pseudo-linear form. Let the matrix functions $A\colon \BBR^n \times \BBR^m \times \BBN \rightarrow \BBR^{n \times n}$, $B\colon \BBR^n \times \BBR^m \times \BBN \rightarrow \BBR^{n \times m}$, and $C\colon \BBR^n \times \BBN \rightarrow \BBR^{p \times n}$ be such that, for all $x \in \BBR^n$, all $u \in \mathcal{U}$, and all $k \in \BBN$,
	\begin{align}
		f(x,u,k) &= A(x,u,k)\,x + B(x,u,k)\,u, \label{eq:scdc_f} \\
		h(x,k) &= C(x,k)\,x. \label{eq:scdc_h}
	\end{align}
	The state- and control-dependent coefficient (SCDC) matrices $A$, $B$, and $C$ represent $f$ and $h$ without truncation. Equation \eqref{eq:scdc_h} requires that, for all $k \in \BBN$, $h(0_n,k) = 0_p$, and, if $0_m \in \mathcal{U}$, then \eqref{eq:scdc_f} requires that, for all $k\in\BBN$, $f(0_n,0_m,k) = 0_n$. Under these conditions, the factorization exists if $f$ is continuously differentiable in $(x,u)$ and $h$ in $x$, and it is in general not unique \cite{cimen2008sdre}. If $0_m \notin \mathcal{U}$, then \eqref{eq:scdc_f} holds, for example, with $A = 0_{n\times n}$ and $B(x,u,k) = f(x,u,k)u^\rmT/\|u\|^2$.
	
	Since the SCDC matrices depend on the unknown state, SCD-MHE evaluates them along a trajectory computed at the preceding iteration, and the resulting quadratic program yields a refined trajectory. Let $\rho \ge 1$ denote the maximum number of iterations at each time step. For all $k \ge \ell$ and all $i \in \{0,\ldots,\rho\}$, let
	\begin{equation}
		\hat{x}_{k|i} \triangleq \big[ \hat{x}_{k,1-\ell|i}^\rmT \ \cdots \ \hat{x}_{k,0|i}^\rmT \big]^\rmT \in \BBR^{n\ell} \label{eq:xhat_stack}
	\end{equation}
	denote the state trajectory computed at iteration $i$, where, for all $j \in \{1-\ell,\ldots,0\}$, $\hat x_{k,j|i} \in \BBR^n$ is the estimate of $x_{k+j}$, and $\hat x_{k|0}$ is the initial trajectory provided by the warm start of Section~\ref{sec:warm}.
	
	For all $k \ge \ell$ and all $\xi \triangleq [\xi_{1-\ell}^\rmT \ \cdots \ \xi_0^\rmT]^\rmT \in \BBR^{n\ell}$, where, for all $j\in\{1-\ell,\ldots,0\}$, $\xi_j \in \BBR^n$, let the SCDC matrices along $\xi$ be defined by
	\begin{gather}
		A_{k,j}(\xi) \triangleq A(\xi_j, u_{k+j}, k+j), \qquad C_{k,j}(\xi) \triangleq C(\xi_j, k+j), \nn\\
		B_{k,j}(\xi) \triangleq B(\xi_j, u_{k+j}, k+j), \label{eq:scdc_along}
	\end{gather}
	and let $\mathcal{Z}_k(\xi) \subset \BBR^{n_z}$ denote the set of all $\zeta \in \BBR^{n_z}$ that satisfy the pseudo-linear equality constraints
	\begin{align}
		\chi_{j+1} & = A_{k,j}(\xi) \chi_j + B_{k,j}(\xi)u_{k+j} + \omega_j, \label{eq:iscd_MHSE_dynamics}\nn \\
		&\qquad\qquad j\in\{1-\ell,\ldots, -1\},\\
		y_{k+j} & = C_{k,j}(\xi)\chi_j + \nu_j, \label{eq:iscd_MHSE_measurement}\nn \\
		&\qquad\qquad j\in\{1-\ell,\ldots, 0\}.
	\end{align}
	For all $k\ge\ell$ and all $i\in\{1,\ldots,\rho\}$, the iterate $\hat z_{k|i} \in \BBR^{n_z}$ is the solution to the quadratic program
	\begin{equation}
		\hat z_{k|i} \triangleq \underset{\zeta \in \mathcal{Z}_k(\hat x_{k|i-1})}{\text{argmin}} \ J_k(\zeta), \label{eq:iscd_MHSE_cost_iter}
	\end{equation}
	where Lemma~\ref{lem:qp_unique} below shows that the minimizer is unique. The iterate is partitioned as in \eqref{eq:zeta}: its state block is $\hat x_{k|i}$, and, for all $j \in \{1-\ell, \dots, -1\}$ and all $j' \in \{1-\ell,\ldots,0\}$, its blocks $\hat{w}_{k,j|i} \in \BBR^n$ and $\hat{v}_{k,j'|i} \in \BBR^p$ are the process disturbance and measurement noise estimates. For brevity, $A_{k,j|i-1} \triangleq A_{k,j}(\hat x_{k|i-1})$, $B_{k,j|i-1} \triangleq B_{k,j}(\hat x_{k|i-1})$, and $C_{k,j|i-1} \triangleq C_{k,j}(\hat x_{k|i-1})$. Since $J_k$ is a convex quadratic function of $\zeta$ with a block-diagonal Hessian, and since \eqref{eq:iscd_MHSE_dynamics}--\eqref{eq:iscd_MHSE_measurement} are linear in $\zeta$ and block bidiagonal in the state columns, \eqref{eq:iscd_MHSE_cost_iter} is a quadratic program whose Karush--Kuhn--Tucker system is sparse with a number of nonzero entries proportional to $\ell$; with a sparse factorization, the cost of each iteration grows linearly with $\ell$. The Hessian is singular, since the intermediate state variables are unpenalized, and Lemma~\ref{lem:qp_unique} shows that the minimizer is unique.
	
	\subsection{Warm Start and Stopping Criteria}
	\label{sec:warm}
	For all $k \ge \ell$, let $i_{*,k}$ denote the final iteration index at time step $k$, which is defined in \eqref{eq:istar} below. For all $k\ge \ell+1$ and all $j \in \{1-\ell, \dots, -1\}$, the initial trajectory is constructed from the final trajectory of time step $k-1$ by the shift \eqref{eq:warm_shift} and the one-step prediction \eqref{eq:warm_predict}, that is,
	\begin{align}
		\hat{x}_{k, j|0} &= \hat{x}_{k-1, j+1|i_{*,k-1}}, \label{eq:warm_shift}\\
		\hat{x}_{k, 0|0} &= f(\hat{x}_{k-1, 0|i_{*,k-1}}, u_{k-1}, k-1). \label{eq:warm_predict}
	\end{align}
	At the first window $k = \ell$, $\hat x_{\ell|0}$ is obtained from a preliminary estimator, such as the EKF, or by forward simulation of \eqref{eq_sys1} from the prior estimate.
	
	For all $k\ge\ell$ and all $i\in\{1,\ldots,\rho\}$, let the trajectory displacement be defined by
	\begin{equation}
		\delta_{k|i} \triangleq \| \hat{x}_{k|i} - \hat{x}_{k|i-1} \|. \label{eq:delta}
	\end{equation}
	Let $\varepsilon \in (0, \infty)$ be a convergence tolerance. For all $k \ge \ell$, the iteration terminates at the final iteration index
	\begin{equation}
		i_{*,k} \triangleq \min\big( \{ i \in \{1,\ldots,\rho\} \colon \delta_{k|i} < \varepsilon\} \cup \{\rho\} \big), \label{eq:istar}
	\end{equation}
	that is, at the first iteration whose displacement falls below $\varepsilon$, or at iteration $\rho$ otherwise. For all $k \ge \ell$, the state estimate at time step $k$ is $\hat x_k \triangleq \hat x_{k,0|i_{*,k}}$.
	
	\subsection{Arrival Cost Update}
	\label{sec:arrival}
	The arrival cost of the first window, $\bar x_1$ and $P_1 \succ 0$, is initialized from the \textit{a priori} state estimate and its covariance. For all $k \ge \ell+1$, the arrival cost vector is the estimate of $x_{k+1-\ell}$ computed at the preceding time step, and the arrival cost covariance is updated by the discrete-time Riccati equation, that is,
	\begin{align}
		\bar{x}_{k+1-\ell} &= \hat{x}_{k-1, 2-\ell | i_{*,k-1}}, \label{eq:xbar_update}\\
		P_{k+1-\ell} &= \bar{A}_{k-\ell} P_{k-\ell} \bar{A}_{k-\ell}^\rmT + Q_{k-\ell}  - \bar{A}_{k-\ell} P_{k-\ell} \bar{C}_{k-\ell}^\rmT\nn \\
		& \quad \cdot \big( \bar{C}_{k-\ell} P_{k-\ell} \bar{C}_{k-\ell}^\rmT + R_{k-\ell} \big)^{-1} \bar{C}_{k-\ell} P_{k-\ell}\bar{A}_{k-\ell}^\rmT ,\label{eq_pk1}
	\end{align}
	where the SCDC matrices at the discarded time step $k-\ell$ are evaluated along the final trajectory of the preceding window, that is, $\bar{A}_{k-\ell} \triangleq A_{k-1,1-\ell|i_{*,k-1}}$ and $\bar{C}_{k-\ell} \triangleq C_{k-1,1-\ell|i_{*,k-1}}$. The recursion \eqref{eq_pk1} is the one-step-predicted covariance recursion of the Kalman filter for the pseudo-linear system $(\bar A_{k-\ell}, \bar C_{k-\ell})$, and thus $P_{k+1-\ell}$ can be interpreted as the covariance of $x_{k+1-\ell}$ given the data up to time step $k-\ell$, whereas $\bar x_{k+1-\ell}$ is the smoothed estimate from the data up to time step $k-1$; \eqref{eq:xbar_update}--\eqref{eq_pk1} is thus a smoothing-type approximation of the arrival cost \cite{Rao2003}. If $P_{k-\ell} \succ 0$, then it follows from the matrix inversion lemma that the sum of the first and third terms in \eqref{eq_pk1} equals $\bar A_{k-\ell}\big(P_{k-\ell}^{-1} + \bar C_{k-\ell}^\rmT R_{k-\ell}^{-1}\bar C_{k-\ell}\big)^{-1}\bar A_{k-\ell}^\rmT \succeq 0$, and thus $P_{k+1-\ell} \succeq Q_{k-\ell} \succ 0$; since $P_1 \succ 0$, it follows by induction that, for all $k \ge \ell$, $P_{k+1-\ell} \succ 0$. Algorithm~\ref{alg:iscd_mhse} summarizes SCD-MHE.
	
	\begin{algorithm}
		\caption{SCD-MHE with Warm Starting}
		\label{alg:iscd_mhse}
		\begin{algorithmic}
			\Require $\ell$, $\rho$, $\varepsilon$, prior arrival cost $\bar{x}_{1}$, $P_{1}$, initial trajectory $\hat x_{\ell|0}$.
			\For{$k = \ell, \ell+1, \dots$}
			\If{$k \ge \ell+1$} compute $\hat x_{k|0}$ by \eqref{eq:warm_shift}--\eqref{eq:warm_predict}.
			\EndIf
			\State $i \gets 0$.
			\Repeat
			\State $i \gets i+1$.
			\State Evaluate $A_{k,j|i-1}$, $B_{k,j|i-1}$, and $C_{k,j|i-1}$ along $\hat x_{k|i-1}$.
			\State Solve the QP \eqref{eq:iscd_MHSE_cost_iter} to obtain $\hat{z}_{k|i}$; extract $\hat x_{k|i}$.
			\State $\delta_{k|i} \gets \| \hat{x}_{k|i} - \hat{x}_{k|i-1} \|$.
			\Until{$\delta_{k|i} < \varepsilon$ or $i = \rho$}
			\State $i_{*,k} \gets i$, \ $\hat{x}_{k} \gets \hat{x}_{k,0|i_{*,k}}$, \ $\bar{x}_{k+2-\ell} \gets \hat{x}_{k,2-\ell|i_{*,k}}$.
			\State Compute $P_{k+2-\ell}$ by \eqref{eq_pk1} with $\bar A_{k+1-\ell}$, $\bar C_{k+1-\ell}$.
			\EndFor
		\end{algorithmic}
	\end{algorithm}

	\section{Theoretical Analysis}
	\label{sec:analysis}
	For all $k \in \BBN$, let the SCDC matrices along the true trajectory be defined by $A_k \triangleq A(x_k, u_k, k)$, $B_k \triangleq B(x_k,u_k,k)$, and $C_k \triangleq C(x_k,k)$. For all integers $t \ge s \ge 0$, let the state transition matrix be defined by $\Phi(t,s) \triangleq A_{t-1} A_{t-2} \cdots A_{s}$ if $t > s$, and $\Phi(t,t) \triangleq I_n$. Furthermore, for all $k \ge \ell$ and all $j \in \{1-\ell,\ldots,0\}$, let $\Phi_{k,j} \triangleq \Phi(k+j,k+1-\ell)$.
	
	\begin{definition}\rm\label{def:observ}
		The system \eqref{eq_sys1}--\eqref{eq_meas1} with the SCDC factorization \eqref{eq:scdc_f}--\eqref{eq:scdc_h} is \textit{uniformly observable} over the horizon $\ell$ if there exists $\alpha \in (0, \infty)$ such that, for all $k\ge \ell$,
		\begin{equation}
			\mathcal{O}_{k} \triangleq \sum_{j=1-\ell}^{0} \Phi_{k,j}^\rmT C_{k+j}^\rmT R_{k+j}^{-1} C_{k+j} \Phi_{k,j} \succeq \alpha I_n. \label{eq:gramian}
		\end{equation}
	\end{definition}
	
	We consider the following assumptions:
	
	\begin{enumerate}[label=(A\arabic*),leftmargin=30pt]
		\item \label{assum:compact}
		The set $\mathcal{U}$ is compact, and there exists a nonempty compact set $\mathcal{X} \subset \BBR^n$ such that, for all $k \in \BBN$, $x_k \in \mathcal{X}$.
		
		\item \label{assum:bounded_noise}
		There exist $\bar{w}, \bar{v} \in (0, \infty)$ such that, for all $k \in \BBN$, $\|w_k\| \le \bar{w}$ and $\|v_k\| \le \bar{v}$.
		
		\item \label{assum:lipschitz}
		There exist $L_A, L_B, L_C \in [0, \infty)$ such that, for all $x, y \in \mathcal{X}$, all $u \in \mathcal{U}$, and all $k \in \BBN$, $\|A(x,u,k) - A(y,u,k)\| \le L_A \|x-y\|$, $\|B(x,u,k) - B(y,u,k)\| \le L_B \|x-y\|$, and $\|C(x,k) - C(y,k)\| \le L_C \|x-y\|$.
		
		\item \label{assum:bounded_A}
		There exist $\bar{a}, \bar{c} \in (0, \infty)$ such that, for all $x \in \mathcal{X}$, all $u \in \mathcal{U}$, and all $k \in \BBN$, $\|A(x,u,k)\| \le \bar{a}$ and $\|C(x,k)\| \le \bar{c}$.
		
		\item \label{assum:observ}
		The system \eqref{eq_sys1}--\eqref{eq_meas1} with the SCDC factorization \eqref{eq:scdc_f}--\eqref{eq:scdc_h} is uniformly observable over the horizon $\ell$.
		
		\item \label{assum:cov_bounds}
		There exist $\underline{q}, \bar q, \underline{r} \in (0,\infty)$ such that, for all $k \in \BBN$, $\underline{q}\, I_n \preceq Q_k \preceq \bar{q}\, I_n$ and $R_k \succeq \underline{r}\, I_p$.
	\end{enumerate}
	
	Assumptions \ref{assum:lipschitz} and \ref{assum:bounded_A} are satisfied if $\mathcal{X}$ is convex and $A$, $B$, and $C$ are continuously differentiable in $x$ and bounded on $\mathcal{X} \times \mathcal{U}$ together with their derivatives, uniformly in $k$. Let $\bar x_{\mathcal{X}} \triangleq \max_{x \in \mathcal{X}} \|x\|$, $\delta_{\mathcal{X}} \triangleq \max_{x,y\in\mathcal{X}}\|x-y\|$, and $\bar u \triangleq \max_{u\in\mathcal{U}}\|u\|$, and define
	\begin{equation}
		\bar\Delta_x \triangleq \delta_{\mathcal{X}}\big(L_A \bar x_{\mathcal{X}} + L_B \bar u\big), \qquad \bar\Delta_y \triangleq \delta_{\mathcal{X}} L_C \bar x_{\mathcal{X}}. \label{eq:Delta_bars}
	\end{equation}
	These constants bound the SCDC mismatch between the true and an estimated trajectory in $\mathcal{X}$, and they vanish if $A$, $B$, and $C$ do not depend on the state.
	
	\subsection{Well-Posedness of the Quadratic Program}
	
	\begin{lemma}\rm\label{lem:qp_unique}
		Let $k \ge \ell$ and $\xi \in \BBR^{n\ell}$. Then, $J_k$ has a unique minimizer on $\mathcal{Z}_k(\xi)$.
	\end{lemma}
	\begin{proof}
		It follows from \eqref{eq:iscd_MHSE_dynamics} and \eqref{eq:iscd_MHSE_measurement} that the map that assigns to each $\zeta \in \mathcal{Z}_k(\xi)$ its state block $\chi$ is a bijection from $\mathcal{Z}_k(\xi)$ onto $\BBR^{n\ell}$ with an affine inverse. Thus, minimizing $J_k$ on $\mathcal{Z}_k(\xi)$ is equivalent to minimizing over $\chi \in \BBR^{n\ell}$ the function $V_k \colon \BBR^{n\ell} \to [0,\infty)$ defined by
		\begin{align}
			V_k(\chi) &\triangleq \|\chi_{1-\ell} - \bar x_{k+1-\ell}\|^2_{P^{-1}_{k+1-\ell}} \nn\\
			&\quad + \sum_{j=1-\ell}^{-1} \|\chi_{j+1} - A_{k,j}(\xi)\chi_j - B_{k,j}(\xi)u_{k+j}\|^2_{Q^{-1}_{k+j}} \nn\\
			&\quad + \sum_{j=1-\ell}^{0} \|y_{k+j} - C_{k,j}(\xi)\chi_j\|^2_{R^{-1}_{k+j}}. \label{eq:Vk}
		\end{align}
		The function $V_k$ is quadratic in $\chi$, and its quadratic part is the form $W_k(\chi) \triangleq \|\chi_{1-\ell}\|^2_{P^{-1}_{k+1-\ell}} + \sum_{j=1-\ell}^{-1} \|\chi_{j+1} - A_{k,j}(\xi)\chi_j\|^2_{Q^{-1}_{k+j}} + \sum_{j=1-\ell}^{0} \|C_{k,j}(\xi)\chi_j\|^2_{R^{-1}_{k+j}}$. Let $\chi \in \BBR^{n\ell}$ be such that $W_k(\chi) = 0$. Since $P_{k+1-\ell} \succ 0$, it follows that $\chi_{1-\ell} = 0_n$. Since, in addition, for all $j \in \{1-\ell,\ldots,-1\}$, $Q_{k+j} \succ 0$, it follows that, for all $j \in \{1-\ell,\ldots,-1\}$, $\chi_{j+1} = A_{k,j}(\xi)\chi_j$, and thus, by induction on $j$, $\chi = 0_{n\ell}$. Therefore, $W_k$ is positive definite, $V_k$ is strictly convex and radially unbounded, and $V_k$ has a unique minimizer on $\BBR^{n\ell}$. Thus, $J_k$ has a unique minimizer on $\mathcal{Z}_k(\xi)$.
	\end{proof}
	
	For all $k \ge \ell$, let $\mathcal{M}_k \colon \BBR^{n\ell} \to \BBR^{n\ell}$ denote the iteration map that assigns to each $\xi \in \BBR^{n\ell}$ the state block of the unique minimizer of $J_k$ on $\mathcal{Z}_k(\xi)$. Then, for all $k \ge \ell$ and all $i \in \{1,\ldots,\rho\}$, $\hat{x}_{k|i} = \mathcal{M}_k(\hat{x}_{k|i-1})$.

	\subsection{Convergence of the Iteration}
	
	For all $k \ge \ell$, let $X_k \triangleq \big[ x_{k+1-\ell}^\rmT \ \cdots \ x_k^\rmT \big]^\rmT \in \BBR^{n\ell}$ denote the true state trajectory over the window.
	
	\begin{lemma}\rm\label{lem:convergence}
		Let $k\ge\ell$ and $\hat x_{k|0} \in \BBR^{n\ell}$, and, for all $i \ge 1$, let $\hat x_{k|i} \triangleq \mathcal{M}_k(\hat x_{k|i-1})$ and $\delta_{k|i} \triangleq \|\hat x_{k|i} - \hat x_{k|i-1}\|$. Assume that there exist $r\in(0,\infty)$ and $L_M \in (0, 1)$ such that, for all $\xi, \xi' \in \bar{\mathcal{N}}_r(X_k)$, $\|\mathcal{M}_k(\xi) - \mathcal{M}_k(\xi')\| \le L_M \|\xi - \xi'\|$, and such that $\|\mathcal{M}_k(X_k) - X_k\| \le (1-L_M)\, r$. Then, the following statements hold:
		\begin{enumerate}[label=(\roman*),leftmargin=20pt]
			\item There exists a unique $\hat x_{k,*} \in \bar{\mathcal{N}}_r(X_k)$ such that $\mathcal{M}_k(\hat x_{k,*}) = \hat x_{k,*}$.
			\item If $\hat x_{k|0} \in \bar{\mathcal{N}}_r(X_k)$, then, for all $i \ge 1$, $\|\hat x_{k|i} - \hat x_{k,*}\| \le L_M^i \|\hat x_{k|0} - \hat x_{k,*}\|$ and $\|\hat x_{k|i} - \hat x_{k,*}\| \le \frac{L_M}{1-L_M}\, \delta_{k|i}$.
			\item The minimizer $\hat z_{k,*}$ of $J_k$ on $\mathcal{Z}_k(\hat x_{k,*})$ satisfies $\hat z_{k,*} \in \mathcal{Z}_k$.
		\end{enumerate}
	\end{lemma}
	\begin{proof}
		Let $\xi \in \bar{\mathcal{N}}_r(X_k)$. It follows from the triangle inequality that $\|\mathcal{M}_k(\xi) - X_k\| \le \|\mathcal{M}_k(\xi) - \mathcal{M}_k(X_k)\| + \|\mathcal{M}_k(X_k) - X_k\| \le L_M r + (1-L_M) r = r$. Thus, $\mathcal{M}_k$ is a contraction that maps the complete metric space $\bar{\mathcal{N}}_r(X_k)$ into itself, and (i) and the first bound in (ii) follow from the Banach fixed-point theorem \cite[Theorem B.1]{khalil2002nonlinear}. To show the second bound in (ii), let $\hat x_{k|0} \in \bar{\mathcal{N}}_r(X_k)$ and $i \ge 1$. Since $\mathcal{M}_k$ maps $\bar{\mathcal{N}}_r(X_k)$ into itself, $\hat x_{k|i-1} \in \bar{\mathcal{N}}_r(X_k)$, and it follows from the contraction property and the triangle inequality that $\|\hat x_{k|i} - \hat x_{k,*}\| = \|\mathcal{M}_k(\hat x_{k|i-1}) - \mathcal{M}_k(\hat x_{k,*})\| \le L_M \big(\delta_{k|i} + \|\hat x_{k|i} - \hat x_{k,*}\|\big)$, which implies the bound.
		
		To show (iii), note that the state block of $\hat z_{k,*}$ is $\mathcal{M}_k(\hat x_{k,*}) = \hat x_{k,*}$. Therefore, the SCDC matrices in the constraints \eqref{eq:iscd_MHSE_dynamics}--\eqref{eq:iscd_MHSE_measurement} satisfied by $\hat z_{k,*}$ are evaluated at the state blocks of $\hat z_{k,*}$ themselves, and \eqref{eq:scdc_f}--\eqref{eq:scdc_h} imply that these constraints coincide with \eqref{eq:MHSE_dynamics}--\eqref{eq:MHSE_measurement}. Thus, $\hat z_{k,*} \in \mathcal{Z}_k$.
	\end{proof}
	
	\begin{remark}\rm\label{rem:fixed_point}
		Statement (iii) shows that the fixed point is feasible for the nonlinear problem \eqref{eq:MHSE_cost}. It is, in general, not a stationary point of \eqref{eq:MHSE_cost}, because the stationarity conditions of the quadratic program involve the SCDC matrices, whereas those of \eqref{eq:MHSE_cost} involve the Jacobians of $f$ and $h$, which, if $A$, $B$, and $C$ are differentiable in $x$, differ from the SCDC matrices by the matrices whose $q$th columns are $(\partial A/\partial x_{(q)})\,x + (\partial B/\partial x_{(q)})\,u$ and $(\partial C/\partial x_{(q)})\,x$, where $x_{(q)}$ is the $q$th component of $x$. The same discrepancy separates the SDRE and Hamilton--Jacobi conditions in SDRE control \cite{cimen2008sdre}, and it vanishes if $A$, $B$, and $C$ do not depend on the state. The contraction condition is a small-gain condition on the sensitivity of the minimizer of \eqref{eq:Vk} to its data \cite[Sec.~12.8]{nocedal2006numerical}, which depends on $L_A$, $L_B$, $L_C$, the magnitudes of the states and inputs, and the relative weights; a common scaling of all weights leaves $\mathcal{M}_k$ unchanged. The second bound in (ii) relates $\delta_{k|i}$, the quantity tested in \eqref{eq:istar}, to the distance to the fixed point.
	\end{remark}

	\subsection{Bounded Estimation Error}
	
	\begin{theorem}\rm\label{thm:bounded_error}
		Consider \eqref{eq_sys1}--\eqref{eq_meas1} with the SCDC factorization \eqref{eq:scdc_f}--\eqref{eq:scdc_h}, and consider Algorithm~\ref{alg:iscd_mhse}. Assume that \ref{assum:compact}--\ref{assum:cov_bounds} are satisfied, and let $\alpha$ be the constant in Definition~\ref{def:observ}. Furthermore, assume that, for all $k \ge \ell$ and all $i \in \{0,\ldots,i_{*,k}\}$, $\hat x_{k|i} \in \mathcal{X}^\ell$. Finally, assume that there exists $\beta \in (0, \infty)$ such that, for all $k \ge \ell$, $\|x_{k+1-\ell} - \bar{x}_{k+1-\ell}\|^2_{P_{k+1-\ell}^{-1}} \le \beta$. Then, there exists $\gamma \in (0,\infty)$ such that, for all $k \ge \ell$, all $i \in \{1,\ldots,i_{*,k}\}$, and all $j \in \{1-\ell,\ldots,0\}$,
		\begin{equation}
			\|x_{k+j} - \hat{x}_{k,j|i}\| \le \gamma. \label{eq:thm_bound}
		\end{equation}
		In particular, for all $k \ge \ell$, $\|x_k - \hat x_k\| \le \gamma$.
	\end{theorem}
	
	\begin{proof}
		Let $k \ge \ell$ and $i \in \{1,\ldots,i_{*,k}\}$, and let $\xi \triangleq \hat x_{k|i-1} \in \mathcal{X}^\ell$. For all $j \in \{1-\ell,\ldots,0\}$, define $e_{j} \triangleq x_{k+j} - \hat x_{k,j|i}$, and define $J_{k|i} \triangleq J_k(\hat z_{k|i}) = \min_{\zeta \in \mathcal{Z}_k(\xi)} J_k(\zeta)$.
		
		Let $\tilde z \in \BBR^{n_z}$ be the decision vector whose state block is $X_k$ and whose disturbance blocks are, for all $j \in \{1-\ell,\ldots,-1\}$, $\tilde\omega_j \triangleq x_{k+j+1} - A_{k,j}(\xi)x_{k+j} - B_{k,j}(\xi)u_{k+j}$, and, for all $j \in \{1-\ell,\ldots,0\}$, $\tilde\nu_j \triangleq y_{k+j} - C_{k,j}(\xi) x_{k+j}$. By construction, $\tilde z \in \mathcal{Z}_k(\xi)$. It follows from \eqref{eq_sys1}, \eqref{eq_meas1}, \eqref{eq:scdc_f}, and \eqref{eq:scdc_h} that, for all $j \in \{1-\ell,\ldots,-1\}$, $\tilde\omega_j = w_{k+j} + (A_{k+j} - A_{k,j}(\xi))x_{k+j} + (B_{k+j} - B_{k,j}(\xi))u_{k+j}$, and, for all $j\in\{1-\ell,\ldots,0\}$, $\tilde \nu_j = v_{k+j} + (C_{k+j} - C_{k,j}(\xi))x_{k+j}$. Since \ref{assum:compact}, \ref{assum:bounded_noise}, and \ref{assum:lipschitz} are satisfied and $\xi \in \mathcal{X}^\ell$, it follows from \eqref{eq:Delta_bars} that, for all $j \in \{1-\ell,\ldots,-1\}$, $\|\tilde\omega_j\| \le \bar w + \bar\Delta_x$, and, for all $j \in \{1-\ell,\ldots,0\}$, $\|\tilde\nu_j\| \le \bar v + \bar \Delta_y$. Since \ref{assum:cov_bounds} is satisfied, for all $j \in \{1-\ell,\ldots,0\}$, $Q_{k+j}^{-1} \preceq \underline{q}^{-1} I_n$ and $R_{k+j}^{-1} \preceq \underline{r}^{-1} I_p$. Therefore, \eqref{eq_JK} and the hypothesis on $\beta$ imply
		\begin{align}
			J_{k|i} \le J_k(\tilde z) \le \bar J &\triangleq \beta + (\ell-1)\,\underline{q}^{-1} (\bar w + \bar\Delta_x)^2 \nn\\
			&\quad + \ell\, \underline{r}^{-1} (\bar v + \bar\Delta_y)^2. \label{eq:Jki_bound}
		\end{align}
		Since each term in \eqref{eq_JK} is nonnegative and \ref{assum:cov_bounds} implies that, for all $j\in\{1-\ell,\ldots,-1\}$, $Q_{k+j}^{-1} \succeq \bar q^{-1} I_n$, it follows that $J_{k|i} \ge \bar q^{-1} \sum_{j=1-\ell}^{-1} \|\hat w_{k,j|i}\|^2$, and thus \eqref{eq:Jki_bound} implies that, for all $j\in\{1-\ell,\ldots,-1\}$,
		\begin{equation}
			\|\hat w_{k,j|i}\| \le \bar\omega \triangleq \sqrt{\bar q\, \bar J}. \label{eq:omega_bar}
		\end{equation}
		
		Since $\hat z_{k|i} \in \mathcal{Z}_k(\xi)$, it follows from \eqref{eq_sys1}, \eqref{eq:scdc_f}, and \eqref{eq:iscd_MHSE_dynamics} that, for all $j \in \{1-\ell,\dots,-1\}$,
		\begin{equation}
			e_{j+1} = A_{k+j}\, e_{j} + \Delta_{x,j} + w_{k+j} - \hat w_{k,j|i}, \label{eq:err_dyn}
		\end{equation}
		where $\Delta_{x,j} \triangleq \big(A_{k+j} - A_{k,j}(\xi)\big)\hat x_{k,j|i} + \big(B_{k+j} - B_{k,j}(\xi)\big)u_{k+j}$. Since \ref{assum:compact} and \ref{assum:lipschitz} are satisfied and $\xi, \hat x_{k|i} \in \mathcal{X}^\ell$, it follows from \eqref{eq:Delta_bars} that, for all $j \in \{1-\ell,\ldots,-1\}$, $\|\Delta_{x,j}\| \le \bar\Delta_x$. Solving \eqref{eq:err_dyn} recursively from $j = 1-\ell$ yields, for all $j \in \{1-\ell,\dots,0\}$,
		\begin{equation}
			e_{j} = \Phi_{k,j}\, e_{1-\ell} + d_{j}, \label{eq:err_expand}
		\end{equation}
		where $d_{j} \triangleq \sum_{s=1-\ell}^{j-1} \Phi(k+j,k+s+1)\big(\Delta_{x,s} + w_{k+s} - \hat w_{k,s|i}\big)$ and $d_{1-\ell} = 0_n$. Since \ref{assum:compact} and \ref{assum:bounded_A} are satisfied, for all $j \in \{1-\ell,\ldots,0\}$ and all $s \in \{1-\ell,\ldots,j-1\}$, $\Phi(k+j,k+s+1)$ and $\Phi_{k,j}$ are products of at most $\ell-1$ matrices of norm at most $\bar a$, and thus $\|\Phi(k+j,k+s+1)\| \le \bar a_\ell$ and $\|\Phi_{k,j}\| \le \bar a_\ell$, where $\bar a_\ell \triangleq \max(1,\bar a)^{\ell-1}$. Since \ref{assum:bounded_noise} is satisfied, it follows from \eqref{eq:omega_bar} that, for all $j \in \{1-\ell,\ldots,0\}$,
		\begin{equation}
			\|d_{j}\| \le \bar d \triangleq (\ell-1)\, \bar a_\ell\, (\bar\Delta_x + \bar w + \bar\omega). \label{eq:d_bar}
		\end{equation}
		
		Since $\hat z_{k|i} \in \mathcal{Z}_k(\xi)$, it follows from \eqref{eq_meas1}, \eqref{eq:scdc_h}, \eqref{eq:iscd_MHSE_measurement}, and \eqref{eq:err_expand} that, for all $j\in\{1-\ell,\ldots,0\}$,
		\begin{equation}
			\hat v_{k,j|i} = C_{k+j}\, e_{j} + \Delta_{y,j} + v_{k+j} = a_j + b_j, \label{eq:vhat}
		\end{equation}
		where $\Delta_{y,j} \triangleq \big(C_{k+j} - C_{k,j}(\xi)\big)\hat x_{k,j|i}$, $a_{j} \triangleq C_{k+j}\, \Phi_{k,j}\, e_{1-\ell}$, and $b_{j} \triangleq C_{k+j}\, d_{j} + \Delta_{y,j} + v_{k+j}$. Since \ref{assum:compact}--\ref{assum:bounded_A} are satisfied, it follows from \eqref{eq:Delta_bars} and \eqref{eq:d_bar} that, for all $j \in \{1-\ell,\ldots,0\}$, $\|\Delta_{y,j}\| \le \bar\Delta_y$ and $\|b_j\| \le \bar c \bar d + \bar\Delta_y + \bar v$.
		
		For all $a, b \in \BBR^p$ and all $W \succ 0$, the Cauchy--Schwarz inequality and $2\|a\|_W\|b\|_W \le \tfrac{1}{2}\|a\|_W^2 + 2\|b\|_W^2$ imply $\|a+b\|_W^2 \ge \tfrac{1}{2}\|a\|_W^2 - \|b\|_W^2$. Applying this inequality to \eqref{eq:vhat} with $W = R_{k+j}^{-1}$ and using \eqref{eq_JK} and \eqref{eq:gramian} yields
		\begin{equation}
			J_{k|i} \ge \sum_{j=1-\ell}^{0} \|\hat v_{k,j|i}\|^2_{R_{k+j}^{-1}} \ge \tfrac{1}{2} e_{1-\ell}^\rmT \mathcal{O}_k e_{1-\ell} - \!\sum_{j=1-\ell}^{0}\! \|b_j\|^2_{R_{k+j}^{-1}}. \label{eq:J_ab}
		\end{equation}
		Since \ref{assum:observ} and \ref{assum:cov_bounds} are satisfied, $\mathcal{O}_k \succeq \alpha I_n$ and, for all $j\in\{1-\ell,\ldots,0\}$, $R_{k+j}^{-1} \preceq \underline{r}^{-1} I_p$, and thus \eqref{eq:Jki_bound} and \eqref{eq:J_ab} imply
		\begin{equation}
			\|e_{1-\ell}\| \le \gamma_0 \triangleq \sqrt{2(\bar{J} + \mathcal{E})/\alpha}, \label{eq:gamma0}
		\end{equation}
		where $\mathcal{E} \triangleq \ell\,\underline{r}^{-1}\big(\bar c \bar d + \bar\Delta_y + \bar v\big)^2$.
		Finally, let $j \in \{1-\ell,\ldots,0\}$. It follows from \eqref{eq:err_expand}, \eqref{eq:d_bar}, and \eqref{eq:gamma0} that $\|x_{k+j} - \hat x_{k,j|i}\| = \|e_j\| \le \bar a_\ell\, \gamma_0 + \bar d \triangleq \gamma$. Since $\gamma$ depends only on $\ell$, $\alpha$, $\beta$, the bounds in \ref{assum:compact}--\ref{assum:cov_bounds}, and \eqref{eq:Delta_bars}, and not on $k$, $i$, or $j$, \eqref{eq:thm_bound} holds. Setting $j = 0$ and $i = i_{*,k}$ yields $\|x_k - \hat x_k\| \le \gamma$.
	\end{proof}
	
	\begin{remark}\rm\label{rem:thm_discussion}
		Theorem~\ref{thm:bounded_error} bounds the error of each iterate, and thus of the estimate returned by Algorithm~\ref{alg:iscd_mhse} for all $\rho$ and $\varepsilon$, without requiring convergence of the iteration. The hypothesis $\hat x_{k|i} \in \mathcal{X}^\ell$ can be enforced by construction if $\mathcal{X}$ is a polytope: appending the linear constraints $\chi_j \in \mathcal{X}$ to \eqref{eq:iscd_MHSE_cost_iter} and projecting $\hat x_{k|0}$ onto $\mathcal{X}^\ell$ preserves the quadratic programming structure, the strict convexity in Lemma~\ref{lem:qp_unique}, and the feasibility of $\tilde z$. The hypothesis on $\beta$ concerns the arrival cost. Since $\bar x_{k+1-\ell}$ is an estimate produced at the preceding time step, this hypothesis can be replaced by a recursive argument that propagates \eqref{eq:thm_bound} across windows and yields a small-gain condition as in \cite{Rao2003, Alessandri2008}; this extension is left for future work.
	\end{remark}
	
	\section{Numerical Results}
	\label{sec:numerical}
	
	We compare SCD-MHE with the EKF, the UKF, and a nonlinear moving-horizon estimator (N-MHE) that solves \eqref{eq:MHSE_cost} on a quadrotor vertical-dynamics benchmark with a saturating rangefinder. Let $x_k \triangleq \begin{bmatrix} z_k & \dot z_k \end{bmatrix}^\mathrm{T}$, where $z_k$ is the altitude and $\dot z_k$ is the vertical velocity. Using forward Euler integration with the sampling period $T_\rms \triangleq 0.05$~s, the dynamics are
	\begin{align}
		z_{k+1} &= z_k + T_\rms \dot z_k + w_{k,(1)}, \\
		\dot z_{k+1} &= \dot z_k + T_\rms \Big( u_k - g - \frac{c_\rmd}{\bar m} \dot z_k |\dot z_k| \Big) + w_{k,(2)},
	\end{align}
	where $\bar m \triangleq 1.5$~kg is the mass, $g \triangleq 9.81$~m/s$^2$ is the gravitational acceleration, $c_\rmd \triangleq 0.25$~kg/m is the drag coefficient, and $u_k \triangleq g + 0.5\sin(k+1)$~m/s$^2$ is the commanded thrust per unit mass, and $\mathcal{U} \triangleq [g-0.5, g+0.5]$. The dynamics are factored as $f(x_k, u_k, k) = A(x_k)x_k + B(u_k)u_k$, where the SCDC matrices depend on $x_k$ and $u_k$ only and are given by
	\begin{equation}
		A(x_k) \triangleq \begin{bmatrix} 1 & T_\rms \\ 0 & 1 - T_\rms \frac{c_\rmd}{\bar m} |\dot z_k| \end{bmatrix}, \quad B(u_k) \triangleq \begin{bmatrix} 0 \\ T_\rms(1 - g/u_k) \end{bmatrix},
	\end{equation}
	where $B$ absorbs the gravitational offset and is defined on $\mathcal{U}$ because $0 \notin \mathcal{U}$. The altitude measurement saturates at $h_{\max} \triangleq 30$~m, that is, $y_k \triangleq h_{\max} \tanh(z_k/h_{\max}) + v_k$, and the SCDC measurement matrix is $C(x_k) \triangleq \begin{bmatrix} c(z_k) & 0 \end{bmatrix}$, where $c(z) \triangleq h_{\max}\tanh(z/h_{\max})/z$ for $z \neq 0$ and $c(0) \triangleq 1$. The function $c$ is continuous, positive, and even, and $c(z) \ge h_{\max}\tanh(1)/|z|$ for all $|z| \ge h_{\max}$, whereas the measurement Jacobian satisfies $\partial h/\partial z = \sech^2(z/h_{\max}) \le 4 e^{-2|z|/h_{\max}}$; at $z = 100$~m, $c(z) \approx 0.30$ and $\sech^2(100/30) \approx 5\times10^{-3}$. Consequently, for $|\xi_{j,(1)}| \ge h_{\max}$, the coefficient $c(\xi_{j,(1)})$ of the altitude variable in \eqref{eq:iscd_MHSE_measurement} is bounded below by $h_{\max}\tanh(1)/|\xi_{j,(1)}|$, whereas the measurement Jacobian at the same point is bounded above by $4e^{-2|\xi_{j,(1)}|/h_{\max}}$.
	
	The process disturbance $w_k$ and the measurement noise $v_k$ are zero-mean Gaussian with covariances $Q \triangleq \diag(10^{-3}, 5\times 10^{-2})$ and $R \triangleq 0.5$; the simulation compares the estimators rather than evaluating the constants of Theorem~\ref{thm:bounded_error}. The true initial state is $x_0 \triangleq \begin{bmatrix} 10 & 0 \end{bmatrix}^\mathrm{T}$. All estimators are initialized with $\hat x_0 \triangleq \begin{bmatrix} 100 & -20 \end{bmatrix}^\mathrm{T}$ and $P_0 \triangleq I_2$. The UKF uses $\alpha_\mathrm{u} \triangleq 10^{-3}$, $\kappa_\mathrm{u} \triangleq 0$, and $\beta_\mathrm{u} \triangleq 2$ \cite{Wan2000}. The SCD-MHE parameters are $\ell \triangleq 12$, $\rho \triangleq 15$, and $\varepsilon \triangleq 10^{-6}$. For $k \in \{1,\ldots,\ell-1\}$, the EKF provides the estimates; for all $j \in \{1-\ell,\ldots,-1\}$, $\hat x_{\ell,j|0}$ is the EKF estimate of $x_{\ell+j}$, and $\hat x_{\ell,0|0} = f(\hat x_{\ell,-1|0}, u_{\ell-1}, \ell-1)$. The prior on $x_0$ serves as the arrival cost of the first window without propagation, that is, $\bar x_1 = \hat x_0$ and $P_1 = P_0$. Each quadratic program is solved by the \texttt{quadprog} interior-point-convex algorithm with dense matrices ($n_z = 58$), and, for numerical conditioning, $10^{-8} I_{n_z}$ is added to the Hessian of $J_k$ and $10^{-5} I_2$ to the arrival cost covariance. The N-MHE is formulated by multiple shooting and solved by the \texttt{fmincon} interior-point algorithm with analytical gradients and constraint Jacobians \cite{rawlings2017model, nocedal2006numerical}, with the same $\ell$, $Q$, $R$, initialization, and regularizations, warm-started with the shifted state trajectory of the preceding step, and with the arrival cost update \eqref{eq_pk1} in which the Jacobians of $f$ and $h$ replace the SCDC matrices. All simulations are executed in MATLAB R2023a on a 3.40~GHz Intel Core i7-13700K processor.
	
	We run 100 Monte Carlo trials of $N \triangleq 120$ steps ($6$~s) each, with $t_k \triangleq k T_\rms$. The root mean squared error (RMSE) is evaluated for $k \ge \ell$, and the per-step execution time of the two moving-horizon estimators is averaged over $k \ge \ell$. Table~\ref{tab:quadrotor_rmse} and Fig.~\ref{fig:quadrotor_tracking} report the results.
	
	\begin{table}[htbp]
		\centering
		\caption{RMSE and Execution Time (Post-Horizon)}
		\label{tab:quadrotor_rmse}
		\begin{tabular}{|l|c|c|c|}
			\hline
			\textbf{Method} & \textbf{Altitude $z$ (m)} & \textbf{Velocity $\dot z$ (m/s)} & \textbf{Time (ms)} \\ \hline
			EKF        & 32.31 & 3.52 & $<0.01$ \\ \hline
			UKF        & 32.34 & 3.51 & $<0.01$ \\ \hline
			N-MHE & 10.26 & 1.95 & 66.16 \\ \hline
			SCD-MHE    & \textbf{0.56}  & \textbf{1.68} & \textbf{1.96} \\ \hline
		\end{tabular}
	\end{table}
	
	At the initial estimated altitude of $100$~m, the measurement Jacobian is $\sech^2(100/30) \approx 5 \times 10^{-3}$, and thus, with $P_0 = I_2$ and $R = 0.5$, the altitude component of the EKF gain is approximately $10^{-2}$. The UKF, whose sigma points lie in the saturated region, behaves in the same way. As shown in Fig.~\ref{fig:quadrotor_tracking}, the EKF and UKF altitude estimates exceed $60$~m for $t \le 1.65$~s and come within $2$~m of the true altitude at $t \approx 2.1$~s; this delay produces their post-horizon altitude RMSE of $32$~m. The N-MHE altitude estimate recovers at $t \approx 0.7$~s, two steps after the first window ($t_\ell = 0.6$~s); its altitude RMSE is $10.26$~m, and its execution time of $66.16$~ms per step exceeds the sampling period $T_\rms = 50$~ms. For the SCD-MHE, the SCDC coefficient $c(100) \approx 0.30$ maps the difference between the measurement $y_k \approx 10$~m and the predicted output $h(100) \approx 29.9$~m into the correction of the altitude estimate, and the iterated quadratic program recovers the true altitude at the first window. The SCD-MHE attains an altitude RMSE of $0.56$~m and a velocity RMSE of $1.68$~m/s at $1.96$~ms per step, which is $34$ times smaller than that of the N-MHE and below the sampling period.
	
	\begin{figure}[htbp]
		\centering
		\includegraphics[width=0.99\columnwidth]{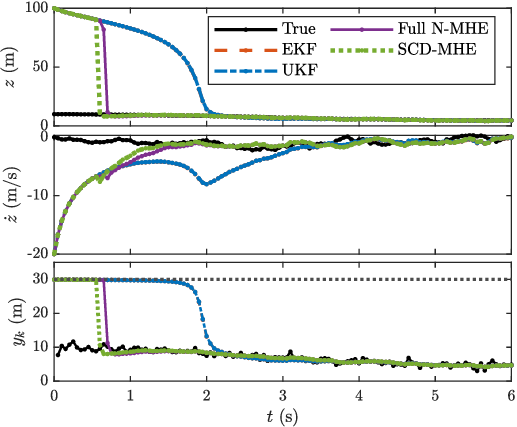}
		\caption{Trajectories of one trial. The EKF and UKF altitude estimates remain in the saturated region of the sensor until $t \approx 1.9$~s; the N-MHE altitude estimate recovers at $t \approx 0.7$~s, and that of the SCD-MHE at the first window ($t = 0.6$~s). The lower panel shows the measured and predicted outputs and the saturation level $h_{\max}$.}
		\label{fig:quadrotor_tracking}
	\end{figure}
	
	
	\section{Conclusion}
	SCD-MHE replaces the nonconvex program of nonlinear MHE by a sequence of quadratic programs that represent the model without truncation and without Jacobians. Each quadratic program has a unique solution, the fixed points of the iteration are feasible for the nonlinear problem, and the estimation error of each iterate is bounded under the hypotheses of Theorem~\ref{thm:bounded_error}.
	
	\bibliographystyle{IEEEtran}
	\bibliography{Ref}
\end{document}